\documentclass[pdflatex,sn-mathphys-num]{sn-jnl}

\usepackage{graphicx}%
\usepackage{multirow}%
\usepackage{amsmath,amssymb,amsfonts}%
\usepackage{amsthm}%
\usepackage{mathrsfs}%
\usepackage[title]{appendix}%
\usepackage{xcolor}%
\usepackage{textcomp}%
\usepackage{manyfoot}%
\usepackage{booktabs}%
\usepackage{algorithm}%
\usepackage{algorithmicx}%
\usepackage{algpseudocode}%
\usepackage{listings}%
\usepackage{array}
\usepackage{url}

\renewcommand{\arraystretch}{2.}

\newtheorem{proposition}{Proposition}

\begin{document}
\title[Article Title]{Optimising Trotter-Suzuki Simulations of Markovian Open Quantum Systems via Classical Search}

\author[1,3]{\fnm{S. M.} \sur{Pillay}}\email{shivani.pillay@nithecs.ac.za}

\author[2,3]{\fnm{I. J.} \sur{David}}\email{ian.david@nithecs.ac.za}

\author*[2,3]{\fnm{I. }\sur{Sinayskiy}}\email{sinayskiy@ukzn.ac.za}

\author[4,5]{\fnm{F. }\sur{Petruccione}}\email{petruccione@sun.ac.za}

\affil[1]{Discipline of Computer Science, School of Agriculture and Science, University of KwaZulu-Natal, Durban, 4001 South Africa}

\affil[2]{Discipline of Physics, School of Agriculture and Science, University of KwaZulu-Natal, Durban 4001, South Africa}

\affil[3]{National Institute for Theoretical and Computational Sciences (NITheCS), KZN Node, University of KwaZulu-Natal, Durban 4001, South Africa}

\affil[4]{School of Data Science and Computational Thinking and Department of Physics, Stellenbosch University, Stellenbosch 7604, South Africa}

\affil[5]{National Institute for Theoretical and Computational Sciences (NITheCS), Stellenbosch, South Africa}

\abstract{Simulating an open quantum system on a digital quantum computer often involves the use of Trotter-Suzuki (TS) Product Formulas (PF) to approximate the system's time evolution. Precise estimates for the required number of Trotter steps (and hence the overall gate count) can be crucial for minimising the computational cost of these methods. Building on established theoretical guarantees, we derive analytic bounds for the First- and Second-Order Deterministic and Randomised TS-PF, directly relating the number of Trotter steps to the model parameters, evolution time and precision. These bounds enable concrete resource estimation for each method. We then present a computationally efficient classical algorithm that uses diamond norm estimates of individual Liouvillian terms and a binary search to significantly reduce the Trotter steps required for a target precision. Our numerical results on two prototypical models—an XX-Spin Chain with boundary driving and local dephasing, and a Transverse-Field Ising Model—show that the theoretical (analytic) bounds are often overly conservative, whereas the empirical (optimised) bounds yield a significantly smaller number of Trotter steps for the same precision. Among the methods investigated, the Second-Order Randomised TS-PF redtypically achieves the lowest resource demands, especially for larger systems. These findings emphasise the significance of empirical bounding strategies in achieving more resource-efficient simulations of Markovian open quantum systems.}

\keywords{Quantum Simulation, Open Quantum Systems, Trotter-Suzuki Product Formulas}

\maketitle

\section{Introduction}\label{sec1}
Simulating open quantum systems (OQS) enables us to understand and predict the behaviour of quantum systems interacting with their environment. This is crucial because real-world systems are rarely fully isolated \cite{breuer2002theory}. Quantum algorithms \cite{lloyd1996universal,childs2018toward,childs2021theory} offer a particularly promising avenue for such simulations by exploiting quantum hardware to circumvent the exponential scaling that plagues classical methods \cite{cleve2016efficient,kamakari2022digital,pocrnic2023quantum}.

In this work, we focus on four methods for the quantum simulation of OQS: the well-established First- and Second-Order Deterministic Trotter-Suzuki (TS) Product Formulas (PF) \cite{suzuki1990fractal,sweke2015universal} and the recently developed First- and Second-Order Randomised TS-PF \cite{david2024faster}. These approaches partition the total evolution time $t$ into $N$ Trotter steps. The TS-PF then approximates the evolution for a small time step $t/N$ using a sequence of simpler exponentials that can be efficiently implemented on a quantum computer. Repeatedly applying this short-time approximation $N$ times to the system's initial state approximates the system's overall time evolution. 

These TS-PF methods come with theoretical bounds that link the precision of the approximation $\epsilon$ to model-specific parameters, the total evolution time and $N$. From these guarantees, asymptotic upper bounds have been derived that describe how $N$ scales with model-specific parameters, the total evolution time and target precision $\epsilon$ \cite{david2024faster}. However, these asymptotic bounds do not yield concrete estimates for $N$. To address this, we derive analytic bounds for each method, which provide explicit estimates for the number of Trotter steps $N$ required to simulate a given model up to a target precision. The gate complexity, which in this work is defined as $N$ multiplied by the number of simpler exponentials that approximate the evolution for time $t/N$, can then be estimated from this analytic bound on $N$. 

Previous studies on TS-PF for simulating closed quantum systems compared analytic and empirical bounds on the number of Trotter steps $N$ required to simulate a Heisenberg model up to a chosen precision \cite{childs2018toward,childs2019faster}. These studies found that the analytic bounds significantly overestimate the $N$ needed to achieve a given target precision. In contrast, the empirical bounds determined through a binary search provided a lower, more accurate estimate of $N$. 

Inspired by this work, we present a classical algorithm for finding the minimal number of Trotter steps $N$ required by the First- and Second-Order Deterministic and Randomised TS-PF to simulate a given model in OQS up to a target precision. This algorithm makes use of the analytic bounds that we have derived and a binary search. 
We then apply this algorithm to find the minimum $N$ required by each method to simulate two prototypical models in open quantum systems. In this way, we establish empirical bounds on $N$ that highlight the extent to which the derived analytic bounds significantly overestimate the actual number of Trotter steps needed.

Using these empirical bounds, we provide a lower, more accurate estimate of the gate complexity required by each method. This leads to more resource-efficient simulations in practice. In addition, we present a comparative analysis of the gate complexities of each method, incorporating model-specific parameters for the two prototypical open quantum systems considered. Overall, these insights pave the way for more resource-efficient quantum simulations of open quantum systems in practice.

\section{Overview of Simulation Methods}\label{sec2}
The state of a Markovian OQS, described by a density matrix $\rho(t)$, evolves according to the Gorini-Kossakowski-Sudarshan-Lindblad (GKSL) equation as \cite{gorini1976completely,lindblad1976generators}

\begin{equation}
    \frac{d}{dt} \rho =  -i[H,\rho] + \sum_{k}\gamma_k(L_k \rho L_k^\dagger - \frac{1}{2}\{L_k^\dagger L_k,\rho\})
\label{gksl_eq}
\end{equation}

\noindent where the Hamiltonian $H=\sum_{s}H_{s}$, the jump operators $\{L_k\}$  and the decay rates $\{\gamma_k\}$ will be specified by the model of the Markovian OQS. For convenience, we  write equation (\ref{gksl_eq}) as  

\begin{equation}
    \frac{d}{dt} \rho = \mathcal{L}(\rho)
\label{shortened_gksl_eq}
\end{equation}

\noindent where the Liouvillian generator $\mathcal{L}$ is defined as 

\begin{equation}
\mathcal{L} = \sum_{k=1}^M \mathcal{L}_k
\label{definition_L}
\end{equation}

\noindent with terms having the form $\mathcal{L}_{k}(\rho) = -i[H_{k},\rho]$ or $\mathcal{L}_{k}(\rho) = \gamma_k(L_k \rho L_k^\dagger - \frac{1}{2}\{L_k^\dagger L_k,\rho\})$. The way we write this generator as a sum is chosen so that we can easily compute or approximate the exponential of each term in (\ref{definition_L}). While this decomposition enables the use of tailored simulation techniques for specific classes of terms, identifying an optimal or effective splitting requires careful consideration. In practice, one aims to partition the generator in a way that best leverages known simulation strategies, as demonstrated in \cite{hagan2023composite}.

The solution to equation (\ref{shortened_gksl_eq}) is 

\begin{equation}
\rho(t) = \Lambda(t) \rho(0),
\label{gksl_sol}
\end{equation}

\noindent where the channel $\Lambda(t)$ is defined as $\Lambda(t) = \mathrm{exp}(t\mathcal{L})= \mathrm{exp}(t\sum_{k=1}^M \mathcal{L}_k)$ and describes the total evolution of the system. Several methods exist to simulate the evolution $\Lambda(t)$ on a quantum computer. This technique is referred to as digital quantum simulation. 

Formally, the goal of digital quantum simulation is to construct an approximation $\tilde{\Lambda} (t)$ of the evolution $\Lambda(t) = \mathrm{exp}(t\sum_{k=1}^M \mathcal{L}_k)$ such that for a given Liouvillian generator $\mathcal{L}$, simulation time $t\ge0$, precision $\epsilon >0$ and distance measure $\mathrm{dist}(\cdot,\cdot)$ 

\begin{equation}
\mathrm{dist}(\tilde{\Lambda} (t),\Lambda(t)) \le \epsilon
\end{equation}

\noindent and $\tilde{\Lambda} (t)$ can be efficiently implemented on a quantum computer. 

In this work, we focus on four methods for the digital quantum simulation of OQS: the well-established First- and Second-Order Deterministic TS-PF and the recently derived First- and Second-Order Randomised TS-PF \cite{david2024faster}. As their names suggest, these methods make use of TS-PF to define a channel $\tilde{\Lambda}(\tau)$ that approximates the evolution for a small time step $\tau = \frac{t}{N}$. The total evolution can then be approximated as 

\begin{equation}
\Lambda(t) \approx (\tilde{\Lambda}(\tau))^N. 
\end{equation}

\noindent The specific choice of $\tilde{\Lambda}(\tau)$ determines the accuracy and efficiency of the simulation. The Deterministic TS-PF define a fixed ordering of exponentials in $\tilde{\Lambda}(\tau)$. The First Order Deterministic TS-PF approximates the evolution for the small time step $\tau$ as 

\begin{equation}
\tilde{\Lambda}_1^{\mathrm{det}}(\tau) = \prod_{k=1}^M \mathrm{exp}(\tau\mathcal{L}_k),      
\end{equation}

\noindent while the Second Order Deterministic TS-PF improves the precision by approximating the evolution for the small time step $\tau$ as  

\begin{equation}
\tilde{\Lambda}_2^{\mathrm{det}}(\tau) = \prod_{k=1}^M \mathrm{exp}(\frac{\tau}{2}\mathcal{L}_k) \prod_{k'=M}^1 \mathrm{exp}(\frac{\tau}{2}\mathcal{L}_{k'}).     
\end{equation}

In contrast, the Randomised TS-PF use sampling techniques to determine the ordering of the exponentials in $\tilde{\Lambda}(\tau)$, which is necessary due to the introduction of convex sums in the product formulas. The First Order Randomised TS-PF approximates the evolution for the small time step $\tau$ as

\begin{equation}
\tilde{\Lambda}_1^{\mathrm{ran}}(\tau) = \frac{1}{2} \prod_{k=1}^M \mathrm{exp}(\tau\mathcal{L}_k) + \frac{1}{2} \prod_{k'=M}^1 \mathrm{exp}(\tau\mathcal{L}_{k'}),
\end{equation}

\noindent while the Second Order Randomised TS-PF approximates the evolution for the small time step $\tau$ as  

\begin{equation}
\tilde{\Lambda}_2^{\mathrm{ran}}(\tau) = \frac{1}{M!} \sum_{\sigma \in \mathrm{Sym}(M)} \tilde{\Lambda}_2^{\sigma}(\tau)
\end{equation}

\noindent where $\sigma$ is a permutation in $\mathrm{Sym}(M)$: the group of all permutations of elements in the set $\{1,\dots,M\}$ and $\tilde{\Lambda}_2^{\sigma}(\tau) = \prod_{k=1}^M \mathrm{exp}(\frac{\tau}{2}\mathcal{L}_{\sigma(k)}) \prod_{k'=M}^1 \mathrm{exp}(\frac{\tau}{2}\mathcal{L}_{\sigma(k')})$. Interestingly, the Second Order Randomised TS-PF represents a convex combination of Second Order Deterministic TS-PF, where each term in the sum corresponds to a deterministic formula with an ordering of exponentials defined by the permutation $\sigma$. 

All of these methods have theoretical guarantees that ensure 

\begin{equation}
    \| \Lambda(t)-\big(\tilde{\Lambda}(\tau)\big)^N\|_\diamond \le \hat\epsilon(t,\lambda,N,M) 
\end{equation}

\noindent where $\|\cdot\|_\diamond$ is the diamond norm, $\hat\epsilon$ is the precision function. This allows us to relate a chosen precision $\epsilon$ to $t$, $\lambda := \mathrm{max} \|\mathcal{L}_k\|_\diamond$, $M$ and $N$ as 

\begin{equation}
    \hat\epsilon(t,\lambda,N,M) \le \epsilon.
\label{epsilon_func_and_epsilon}
\end{equation}

It is important to note that, in practice, the implementation of randomised formulas involves classical random sampling or quantum forking \cite{david2024faster}. However, since the randomised formulas can be viewed as convex mixtures of deterministic product formulas (PFs), we are able to find the precision function $\hat\epsilon$ without requiring any statistical analysis. This is because the sampling is performed only during the actual simulation of a given generator and does not impact the theoretical analysis of the precision.

\noindent Table \ref{tab:ts_formulas} shows the specific precision functions for each method.

\begin{table}[h!]
\centering
\renewcommand{\arraystretch}{2} 
\caption{Error functions and their corresponding TS-PF, as well as the formulas for how one computes the gate complexity in terms of the number of Trotter steps $N$. This is usually just the number of exponentials in the product formula times $N$. The deterministic error functions are sourced from \cite{suzuki1990fractal} while randomised error functions are sourced from \cite{david2024faster}.}
\begin{tabular}{|c|c|c|}
\hline
\textbf{TS-PF} & \textbf{Precision Function} & \textbf{Gate Complexity} \\ 
\hline
$\tilde{\Lambda}_{1}^{\mathrm{det}}(t/N)^{N}$ & $\dfrac{(t\lambda M)^{2}}{N}\exp\left(\dfrac{t\lambda M}{N}\right)$ & $MN$ \\[15pt] 
$\tilde{\Lambda}_{1}^{\mathrm{ran}}(t/N)^{N}$ & $\dfrac{(t\lambda M)^{3}}{3N^{2}}\exp\left(\dfrac{t\lambda M}{N}\right)$ & $MN$ \\[15pt]
$\tilde{\Lambda}_{2}^{\mathrm{det}}(t/N)^{N}$ & $\dfrac{(t\lambda M)^{3}}{3N^{2}}\exp\left(\dfrac{t\lambda M}{N}\right)$ & $2MN$\\[15pt]
$\tilde{\Lambda}_{2}^{\mathrm{ran}}(t/N)^{N}$ & $\dfrac{(t\lambda)^{3}M^{2}}{N^{2}}\exp\left(\dfrac{t\lambda M}{N}\right)$& $2MN$ \\[5pt]
\hline
\end{tabular}

\label{tab:ts_formulas}
\end{table}

\section{Derivation of Analytic Bounds on the Number of Trotter Steps}\label{sec_analytic_bounds}

Similarly to \cite{childs2018toward}, we derive analytic bounds on the number of Trotter steps needed for each method. The following propositions define the analytic upper bounds on $N$ that we have derived for the First- and Second-Order Deterministic and Randomised TS-PF. These bounds allow us to directly relate $N$ to $t$, $\lambda := \max_{k}\{ \|\mathcal{L}_k\|_\diamond\}$, $M$ and the chosen precision $\epsilon$. To establish each bound, we apply equation (\ref{epsilon_func_and_epsilon}) which relates the precision $\epsilon$ to the precision function $\hat\epsilon(t,\lambda,N,M)$. For the Deterministic and Randomized TS-PF, we employ the specific forms of $\hat\epsilon(t,\lambda,N,M)$ derived in \cite{suzuki1990fractal} and \cite{david2024faster}, respectively. We then control the exponential factor in the precision function. A detailed proof is provided for the First-Order Deterministic TS-PF while the other proofs, which follow the same logic, are omitted. 

\begin{proposition}[First-Order Deterministic TS-PF]
Let $t \geq 0$ be the total evolution time, $N \in \mathbb{N}$ be the number of Trotter steps, and $\epsilon > 0$ be the chosen precision. Given the First-Order Deterministic TS-PF,
\begin{align}
\tilde{\Lambda}_{1}^{\text{det}}(\tau) = \prod_{k=1}^{M} e^{\tau \mathcal{L}_{k}},
\end{align}
where $\tau \geq 0$, the First-Order Deterministic analytic bound on $N$ is given by
\begin{align}
N_{1,\text{det}}^{\mathrm{analytic}} := \left\lceil \max \left\{ t \lambda M, \frac{e (t \lambda M)^{2}}{\epsilon} \right\} \right\rceil.
\end{align}
\end{proposition}

\begin{proof}
Using the precision function from Table \ref{tab:ts_formulas} and choosing a precision $\epsilon>0$, we have the following bound:
\begin{align}
\label{first_ord_det_error}
    \|\Lambda(t)-\tilde{\Lambda}(t/N)^{N}\|_{\diamond}\leq \frac{(t\lambda M)^{2}}{N}\exp\left(\dfrac{t\lambda M}{N}\right)\leq \epsilon.
\end{align}
We need to control the exponential factor in the precision function to obtain an analytic bound for $N$ in terms of $\lambda$, $t$, $M$ and $\epsilon$. We start by imposing the constraint that
\begin{align}
\label{N_1}
    \frac{t\lambda M}{N}\leq 1. 
\end{align}
This implies that we must choose $N \geq t\lambda M$, which in turn implies that the exponential term is upper-bounded by $e$. Then, in the inequality (\ref{first_ord_det_error}) we have
\begin{align}
\label{N_2}
    \frac{e(t\lambda M)^{2}}{\epsilon}& \leq N.
\end{align}
Finally, since $N$ should be a natural number and needs to satisfy both inequalities (\ref{N_1}) and (\ref{N_2}), we have 
\begin{align}
    N_{1,\text{det}}^{\text{analytic}} =\left \lceil \max\left\{t\lambda M,  \frac{e(t\lambda M)^{2}}{\epsilon} \right\}\right\rceil.
\end{align}
\end{proof}

\begin{proposition}[First-Order Randomised TS-PF]
Let $t \geq 0$ be the total evolution time, $N \in \mathbb{N}$ be the number of Trotter steps, and $\epsilon > 0$ be the chosen precision. Given the First-Order Randomised TS-PF,
\begin{align}
\tilde{\Lambda}_{1}^{\text{ran}}(\tau) = \frac{1}{2}\prod_{k=1}^{M} e^{\tau \mathcal{L}_{k}}+\frac{1}{2}\prod_{k'=M}^{1} e^{\tau \mathcal{L}_{k'}},
\end{align}
where $\tau \geq 0$, the First-Order Randomised analytic bound on $N$ is given by
\begin{align}
N_{1,\text{ran}}^{\mathrm{analytic}} := \left\lceil \max \left\{ t \lambda M, \sqrt{\frac{e (t \lambda M)^{3}}{3\epsilon}} \right\} \right\rceil.
\end{align}
\end{proposition}

\begin{proposition}[Second-Order Deterministic TS-PF]
Let $t \geq 0$ be the total evolution time, $N \in \mathbb{N}$ be the number of Trotter steps, and $\epsilon > 0$ be the chosen precision. Given the Second-Order Deterministic TS-PF,
\begin{align}
\tilde{\Lambda}_{2}^{\text{det}}(\tau) = \prod_{k=1}^{M} e^{\frac{\tau}{2} \mathcal{L}_{k}}\prod_{k'=M}^{1} e^{\frac{\tau}{2} \mathcal{L}_{k'}},
\end{align}
where $\tau \geq 0$, the Second Order Deterministic analytic bound on $N$ is given by
\begin{align}
N_{2,\text{det}}^{\mathrm{analytic}} := \left\lceil \max \left\{ t \lambda M, \sqrt{\frac{e (t \lambda M)^{3}}{3\epsilon}} \right\} \right\rceil.
\end{align}
\end{proposition}

\begin{proposition}[Second-Order Randomised TS-PF]
Let $t \geq 0$ be the total evolution time, $N \in \mathbb{N}$ be the number of Trotter steps, and $\epsilon > 0$ be the chosen precision. Given the Second-Order Randomised TS-PF,
\begin{align}
\tilde{\Lambda}_2^{\mathrm{ran}}(\tau) = \frac{1}{M!} \sum_{\sigma \in \mathrm{Sym}(M)}\prod_{k=1}^M \mathrm{exp}(\frac{\tau}{2}\mathcal{L}_{\sigma(k)}) \prod_{k'=M}^1 \mathrm{exp}(\frac{\tau}{2}\mathcal{L}_{\sigma(k')})
\end{align}
where $\sigma$ is a permutation in $\mathrm{Sym}(M)$ and $\tau \geq 0$, the Second Order Randomised analytic bound on $N$ is given by
\begin{align}
N_{2,\text{ran}}^{\mathrm{analytic}} := \left\lceil \max \left\{ t \lambda M, \sqrt{\frac{e (t \lambda)^{3}M^{2}}{\epsilon}} \right\} \right\rceil.
\end{align}
\end{proposition}

\section{Formulation of the Search Problem}\label{sec_binary_search}

The total native gate count required by a specific method to simulate a given model to a target precision $\epsilon$ depends on three parameters: (i) the number of Trotter steps $N$, (ii) the number of exponentials per step in $\tilde{\Lambda}(\tau)$ and (iii) the native gates required per exponential. While the second parameter is specified by the simulation method, the third is method-independent, depending instead on the model's Liouvillian terms and the chosen dilation \cite{stinespring1955positive, szHokefalvi1954contractions} and decomposition \cite{dawson2005solovay, li2013decomposition} schemes. One way the native gate count can be minimised is through the reduction of the number of Trotter steps $N$. In this work, we define gate complexity using only the first two parameters. This approach allows us to quantify the reduction in gate count achieved by optimising $N$ for a given model, while ensuring a fair comparison of the relative gate efficiency across different simulation methods for that same model.

Having established that gate complexity can be minimized through $N$, we now describe the procedure used to identify its optimal value for a given model, simulation method and target precision $\epsilon$. This minimum $N$ is located via a binary search \cite{knuth1997art} within the interval $[1, N_{\mathrm{upper}}]$. The upper bound $N_{\mathrm{upper}}$ is determined by initially setting $N = 1$ and doubling it until $\hat\epsilon(t, \lambda, N, M) \leq \epsilon$.

The binary search algorithm works as follows: we begin by evaluating the precision function $\hat\epsilon(t,\lambda,N,M)$ at the middle value of $N$ in $[1,N_{\mathrm{upper}}]$. If the resulting precision is not equal to the target precision, we ignore half the values of $N$ in $[1,N_{\mathrm{upper}}]$ where the minimum $N$ cannot lie. The search continues on the remaining half, with the precision function evaluated at the middle value of the new interval $[N_{\mathrm{lower}}, N_{\mathrm{upper}}]$. This process repeats until the smallest value of $N$ that achieves the precision closest to the target precision is found.

A detailed implementation of the procedure for finding the initial $N_{\mathrm{upper}}$ and the binary search algorithm is provided in Algorithm~\ref{binary_search} (see Appendix~A).

By using binary search, we can find the minimum $N$ in $[1, N_{\mathrm{upper}}]$ with, on average, $\mathcal{O}(\log n)$ iterations, where $n$ is the initial number of values in $[1, N_{\mathrm{upper}}]$. Furthermore, since the precision function $\hat\epsilon(t, \lambda, N, M)$ is only evaluated as needed, the number of evaluations performed is also $\mathcal{O}(\log n)$ on average.

\section{Results and Discussion}\label{sec4}

In this section, we present the two models that will be used in our empirical study: an XX-Spin Chain with Boundary Driving and Local Dephasing \cite{vznidarivc2015relaxation}, and a Transverse Field Ising Model with Local Dephasing \cite{hashizume2022dynamical}. We first detail each system's Hamiltonian and dissipative dynamics, including their respective Liouvillian generators. Next, we outline the methodology for computing the diamond norm of the generator terms, leveraging a bound to circumvent the need to solve optimization problems \cite{nechita2018almost}. Finally, we present and discuss the empirical and analytic bounds on the number of Trotter steps required by the four simulation methods outlined in section \ref{sec2}. Estimated using the empirical and analytic bounds, the gate complexities are also presented and discussed. 

\subsection{XX-Spin Chain with Boundary Driving and Local Dissipation}

This model consists of spins arranged on a chain, with the boundaries coupled to two heat baths and local dephasing occurring at each site. The following Hamiltonian governs the system:

\begin{align}
H=\sum_{j=1}^{P-1}X_{j}X_{j+1}+Y_{j}Y_{j+1},
\end{align}

\noindent where $X_{j}$ and $Y_{j}$ are the Pauli matrices acting on the $j$-th spin and $P$ is the number of spins. This model has two types of dissipative terms: boundary terms, which account for the driving, and local dephasing terms for each spin. The jump operators for the boundary driving terms are:

\begin{align}
L_{1}=\sqrt{\frac{\Omega}{2}} \ \sigma_{1}^{+}, && L_{2}=\sqrt{\frac{\Omega}{2}} \ \sigma_{1}^{-}, && L_{3}=\sqrt{\frac{\Omega}{2}} \ \sigma_{P}^{+}, && L_{4}=\sqrt{\frac{\Omega}{2}} \ \sigma_{P}^{-},
\end{align}

\noindent where $\sigma_{k}^{\pm}=(X_{k}\pm iY_{k})/2$ for $k=1,...,P$ and the parameter $\Omega$ is the coupling strength of the heat baths at the boundaries. $\Omega$ has a typical value in the interval $[0.1, 10]$. For our study, we choose $\Omega=3.94$. The jump operators for the local dephasing terms are:

\begin{align}
L_{j}^{\text{deph}}=\sqrt{\frac{\gamma}{2}}Z_{j}
\end{align}

\noindent where $\gamma$ is the dephasing strength and should be in the range $[0,1]$. In this work, we use $\gamma=0.31$. 

\noindent The Liouvillian generator for this model has the form:

\begin{align}
\mathcal{L}(\rho)=&\sum_{j=1}^{P-1}-i\left[X_{j}X_{j+1},\rho \right]+\sum_{j=1}^{P-1}-i\left[Y_{j}Y_{j+1},\rho \right]\nonumber\\
+\sum_{j=1}^{4}&\left(2L_{j}\rho L_{j}^{\dagger}-\{L_{j}^{\dagger}L_{j},\rho\}\right)+\gamma 
\sum_{k=1}^{P}\left(Z_{k}\rho Z_{k}-\rho\right).
\end{align}

\noindent We observe that the number of the terms in this generator is $M=2P+3$, where we group the $-i[X_{j}X_{j+1},\rho]$ and $-i[Y_{j}Y_{j+1},\rho]$ commutator terms at each site $j$ into a single Liouvillian term, giving $(P-1)$ combined coupling terms, $4$ boundary driving terms, and $P$ local dephasing terms. 

\subsection{Transverse Field Ising Model with Local Dephasing}

This model has spins on a lattice with neighbour-neighbour coupling and local dephasing at each spin site. We consider finite TFIM instances with $n=2,3,4,5,$ and $6$ spins. Each instance is specified by an interaction graph $G_n=(V_n,E_n)$, where the vertices represent spins and the edges represent nearest-neighbour $ZZ$ interactions. The interaction graphs are shown in Figure \ref{fig:ising-lattices}. For $n=2$ we use a single bond, for $n=3$ a three-site chain, for $n=4$ a $2\times 2$ square, for $n=5$ a five-site ring, and for $n=6$ a $2\times 3$ rectangular lattice. The Hamiltonian for this model is

\begin{align}
    H= -J\sum_{\langle i,j\rangle}Z_{i}Z_{j}-h\sum_{j}X_{j},
\end{align}

\noindent where $J$ is the coupling strength between neighbouring spins and $h$ is the transverse field strength for a field along the $x-$direction. The $\langle i,j\rangle$ represents summation over nearest-neighbour interaction pairs, and it ensures we do not double count. In this work, we use $J=1$ and $h=0.5$. The jump operator for the local dephasing is 

\begin{align}
    L_{j}=\sqrt{\gamma} \ Z_{j},
\end{align}

\noindent where $\gamma$ is the dephasing strength and we use $\gamma =0.1$. 

\noindent The Liouvillian generator for this model is

\begin{align}
    \mathcal{L}(\rho) = iJ\sum_{\langle i,j \rangle}[Z_{i}Z_{j},\rho]+ih\sum_{j}[X_{j},\rho]+\gamma\sum_{j}(Z_{j}\rho Z_{j}-\rho).
\end{align}

\noindent The number of Liouvillian terms is determined by the number of edges in the interaction graph. Each edge contributes one two-site $ZZ$ commutator term, while each spin contributes one transverse-field term and one local dephasing term. Therefore, $M=|E_n|+2n$. For the five TFIM configurations used here, this gives $M=5,8,12,15,$ and $19$ for $n=2,3,4,5,$ and $6$, respectively.

\begin{figure}[h!]
    \centering
    \includegraphics[width=\linewidth]{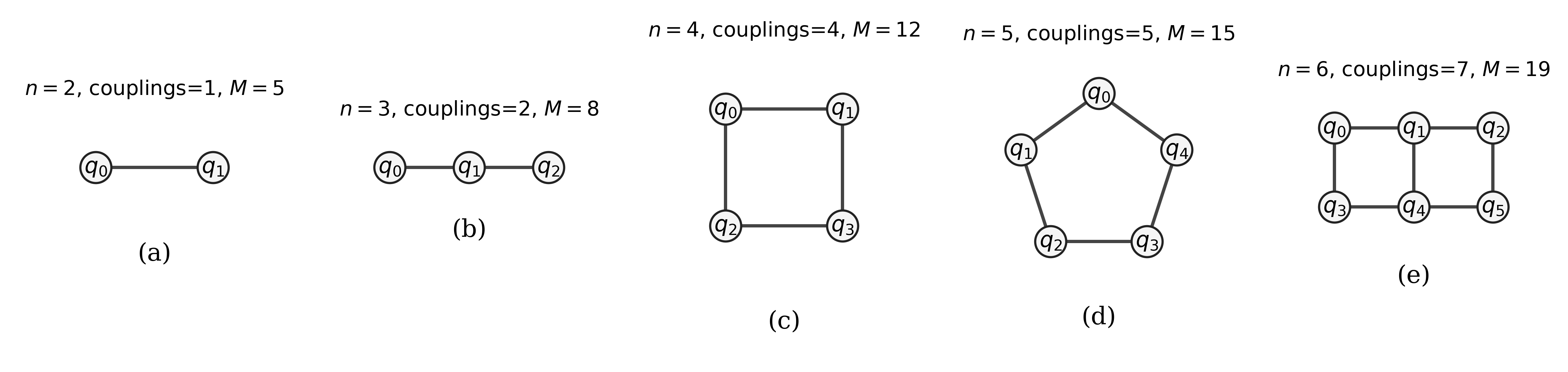}
    \caption{The TFIM interaction graphs used for $n=2,3,4,5,$ and $6$ spins. Nodes represent spins and edges represent nearest-neighbour $ZZ$ interactions.}
    \label{fig:ising-lattices}
\end{figure}

\subsection{Computing the Diamond Norm of the Generator}
For both models, we need to compute $\lambda = \max_{k}\{\|\mathcal{L}_{k}\|_{\diamond}\}$. This involves computing the diamond norm of each term in the Liouvillian generator $\mathcal{L}$. Typically, the diamond norm is estimated by solving an optimization problem. However, this becomes more computationally intensive as the size of the system increases. To circumvent this, we use a bound on the diamond norm derived in \cite{nechita2018almost}. Under certain conditions, this bound is equivalent to the diamond norm itself. The bound is stated and proven in proposition 1 of \cite{nechita2018almost} but we state it here for convenience. Given a linear map $\Phi : \mathcal{M}_{d}(\mathbb{C}) \rightarrow \mathcal{M}_{d}(\mathbb{C})$, we have

\begin{align}
    \|\Phi\|_{\diamond}\leq \frac{\|\mathrm{tr}_{2}\sqrt{J(\Phi)^{\dagger}J(\Phi)}\|_{\infty}+\|\mathrm{tr}_{2}\sqrt{J(\Phi)J(\Phi)^{\dagger}}\|_{\infty}}{2},
\end{align}

\noindent where $J(\Phi)$ is the Choi-Jamiolkowski matrix \cite{choi1975completely}, $\mathrm{tr}_{2}$ represents the partial trace over the second tensor factor and $\|\cdot\|_{\infty}$ is the infinity norm for matrices. We see that this inequality has equality if $\mathrm{tr}_{2}\sqrt{J(\Phi)^{\dagger}J(\Phi)}$ and $\mathrm{tr}_{2}\sqrt{J(\Phi)J(\Phi)^{\dagger}}$ are scalar. 

While the Choi-Jamiolkowski construction scales exponentially with total system size $P$ — requiring a $4^P \times 4^P$ matrix representation \cite{choi1975completely} — this overhead is circumvented in the present work by exploiting the locality of the Liouvillian generator. Because our models consist of terms $\mathcal{L}_{k}$ acting on at most $k=2$ qubits, each term is effectively a local superoperator tensored with the identity. Since the diamond norm is stable under tensoring with identity maps, we compute the norm only on the local support. Consequently, the Choi-Jamiolkowski matrices utilised here remain of fixed dimension $(4^k \times 4^k)$, and the total classical overhead scales linearly with the number of terms rather than the Hilbert-space dimension. This approach is consistent with established methods for efficient Markovian dynamics simulation \cite{BarthelKliesch2012}. While a simulation involving genuinely non-local terms (where $k \sim P$) would require analytic norm bounds or structure-exploiting approximations, the explicit Choi-Jamiolkowski construction provides an exact and efficient metric for the local models considered in this work.

Therefore, for each $\mathcal{L}_{k}$, we can compute this bound and check if the matrices $\mathrm{tr}_{2}\sqrt{J(\mathcal{L}_{k})^{\dagger}J(\mathcal{L}_{k})}$ and $\mathrm{tr}_{2}\sqrt{J(\mathcal{L}_{k})J(\mathcal{L}_{k})^{\dagger}}$ are scalar. If the matrices are scalar, we have obtained the exact diamond norm of $\mathcal{L}_{k}$. Otherwise, we have obtained an upper bound for the diamond norm of $\mathcal{L}_{k}$. In this way, we are able to compute or bound $\lambda$. 

\subsection{Comparison of Analytic Bounds and Empirical Bounds}

In this section, we compare the empirical and analytic bounds on the number of Trotter steps required by each simulation method for each model. We also compare the gate complexities estimated for the methods from the number of Trotter steps using the formulas in Table \ref{tab:ts_formulas}.

\begin{figure}[h!]
    \centering
    \includegraphics[width=\linewidth]{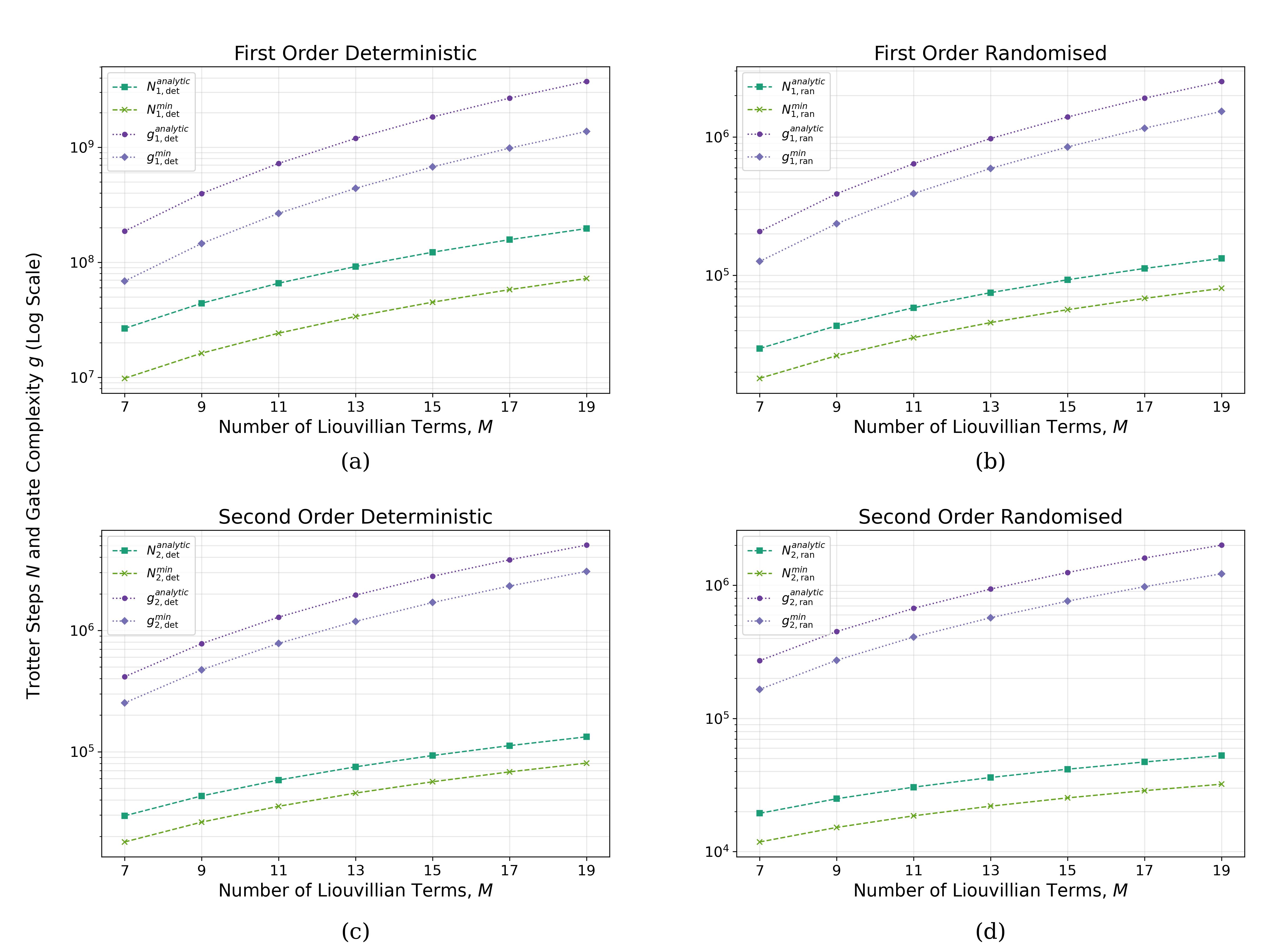}
    \caption{The number of Trotter steps, $N$, and gate complexities, $g$, required by each method for the XX-Spin Chain model with varying $M$, the number of terms in the model's generator. $N^{\mathrm{analytic}}$ for each method is obtained from the analytic bounds derived in Section \ref{sec_analytic_bounds} and $g^{\mathrm{analytic}}$ is obtained using these values of $N$ and the formulas in Table \ref{tab:ts_formulas}. $N^{\mathrm{min}}$ for each method is obtained from the binary search outlined in Section \ref{sec_binary_search} and $g^{\mathrm{min}}$ is obtained using these values of $N$ and the formulas in Table \ref{tab:ts_formulas}. (a)-(b) shows $N^{\mathrm{analytic}}$, $N^{\mathrm{min}}$, $g^{\mathrm{analytic}}$ and $g^{\mathrm{min}}$ for  the First Order Deterministic and Randomised Trotter-Suzuki Product Formulas. (c)-(d) shows $N^{\mathrm{analytic}}$, $N^{\mathrm{min}}$, $g^{\mathrm{analytic}}$ and $g^{\mathrm{min}}$ for the Second Order Deterministic and Randomised Trotter-Suzuki Product Formulas. All $y$-axes are displayed on a logarithmic scale.}
    \label{fig:XX-spin-chain}
\end{figure}

Fig. \ref{fig:XX-spin-chain} and Fig. \ref{fig:TFIM-2D} show both the number of Trotter steps and the gate complexity required by each simulation method for the XX-Spin Chain Model and the Transverse Field Ising Model, respectively. For the XX-Spin Chain Model, we found $\lambda=7.071$ and we used $\epsilon=10^{-3}$, $t=2$. For the Transverse Field Ising Model, we found $\lambda=8.00$ (the same for all $n$) and we used $\epsilon=10^{-5}$, $t=5$. As expected, as $M$ increases for both models, the Randomised Product Formulas require fewer Trotter steps than the Deterministic Product Formulas. This indicates that the Randomised Product Formulas perform better as $M$ increases, that is, for larger systems. 

We also observe that, for both models and for every method, the empirical bounds obtained for $N$ are significantly lower than the analytic bounds. The empirical bounds reveal that fewer Trotter steps are needed to obtain the same precision as the analytic bound estimates. 

\begin{figure}[h!]
    \centering
    \includegraphics[width=\linewidth]{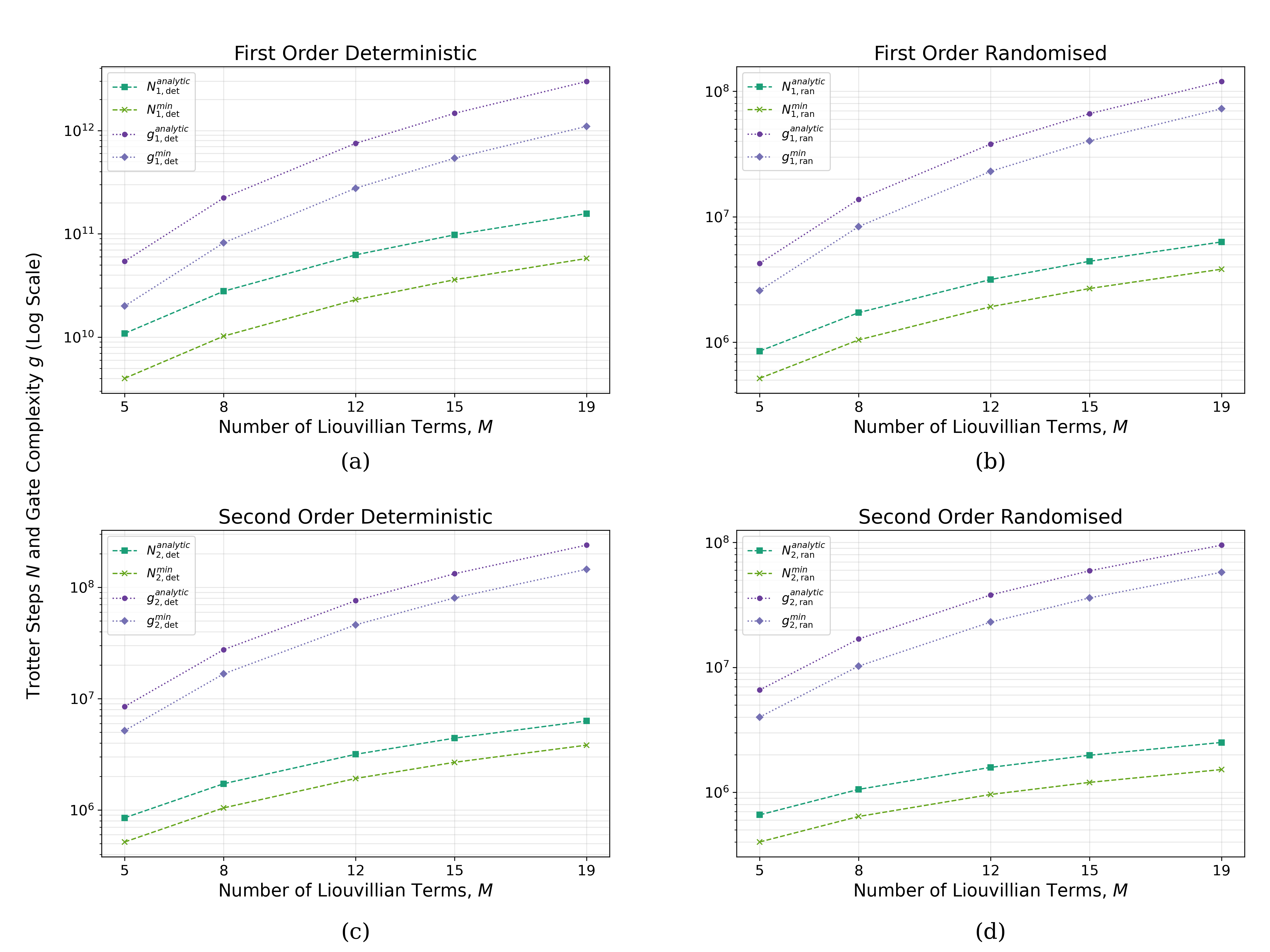}
    \caption{The number of Trotter steps, $N$, and gate complexities, $g$, required by each method for the Transverse Field Ising model with varying $M$, the number of terms in the model's generator. For this model, $M=5$ corresponds to $n=2$ spins, $M=8$ corresponds to $n=3$ spins, $M=12$ corresponds to $n=4$ spins, $M=15$ corresponds to $n=5$ spins and $M=19$ corresponds to $n=6$ spins. $N^{\mathrm{analytic}}$ for each method is obtained from the analytic bounds derived in Section \ref{sec_analytic_bounds} and $g^{\mathrm{analytic}}$ is obtained using these values of $N$ and the formulas in Table \ref{tab:ts_formulas}. $N^{\mathrm{min}}$ for each method is obtained from the binary search outlined in Section \ref{sec_binary_search} and $g^{\mathrm{min}}$ is obtained using these values of $N$ and the formulas in Table \ref{tab:ts_formulas}. (a)-(b) shows $N^{\mathrm{analytic}}$, $N^{\mathrm{min}}$, $g^{\mathrm{analytic}}$ and $g^{\mathrm{min}}$ for the First Order Deterministic and Randomised Trotter-Suzuki Product Formulas. (c)-(d) shows $N^{\mathrm{analytic}}$, $N^{\mathrm{min}}$, $g^{\mathrm{analytic}}$ and $g^{\mathrm{min}}$ for the Second Order Deterministic and Randomised Trotter-Suzuki Product Formulas. All $y$-axes are displayed on a logarithmic scale.}
    \label{fig:TFIM-2D}
\end{figure}
 
For example, for the XX-Spin Chain Model and for $M=15$, we see that the difference in $N$ between the empirical and analytic bound for the First-Order Deterministic TS-PF is of the order $10^{7}$. For the same model for $M=15$, the difference between the analytic and empirical gate complexity is of order $10^{9}$. If we look at the Transverse Field Ising Model, we see similar trends in the gate complexities. For this model for $M=19$, the difference between the analytic and the empirical gate complexity for the First-Order Deterministic Formula is approximately $1.89\times 10^{12}$. 

It is also important to note that, although the First-Order Randomised TS-PF and the Second-Order Deterministic TS-PF require the same number of Trotter steps to achieve a given target precision, their gate complexities differ significantly. Specifically, the gate complexity of the Second-Order Deterministic TS-PF is 2MN, while that of the First-Order Randomised TS-PF is MN, where M is the number of terms in the Liouvillian and N is the number of Trotter steps. As a result, the First-Order Randomised TS-PF uses half as many gates as the Second-Order Deterministic TS-PF to reach the same target precision.

Lastly, for both models, we see that the Second-Order Randomised Formula requires the fewest number of Trotter steps and therefore has the best gate complexity even as the system size increases. Both the analytic and empirical bounds indicate this.

Regarding physical resource requirements and native gate complexity, we note that Randomised TS-PF can be implemented via either classical random sampling or quantum forking \cite{david2024faster}. In the classical sampling regime, where a permutation is drawn classically
before each Trotter step, no ancilla qubits or control-select logic are required, and the gate
complexity advantage reported here is fully realised on hardware. If quantum forking is used instead,
the additional ancilla overhead may partially offset the reduction in $N$; the extent to which the
advantage persists in this regime is hardware-dependent and left for future investigation.

\section{Conclusion}\label{sec5}

In this work, we formulate analytic bounds for the First- and Second-Order Deterministic and Randomised Trotter-Suzuki Product Formulas (TS-PF),based on \cite{sweke2015universal, david2024faster}, that directly link the number of Trotter steps N to model parameters, total evolution time, and target precision. We complemented these results with empirical bounds—obtained via a binary search—to estimate $N$ for simulating two representative open quantum systems. Our findings show that the empirical bounds substantially tighten the analytic ones, indicating that far fewer Trotter steps are needed than the analytic estimates suggest. This reduction in $N$ directly translates into lower gate complexity and hence improves the overall efficiency of quantum simulations in practice. Furthermore, we observed that the Randomised TS-PF generally outperforms the Deterministic approaches for larger system sizes, primarily due to more favourable scaling with the number of generator terms. Among all the methods evaluated, the Second-Order Randomised TS-PF achieved the smallest N and, consequently, the lowest gate complexity. These results underscore the value of empirical bounding methods in quantum simulation. Future work could focus on deriving tighter analytic bounds that better align with empirical observations and on extending comparative studies to other quantum algorithms used for open-system simulations \cite{cleve2016efficient,childs2017efficient,li2022simulating,pocrnic2023quantum}, ultimately enabling more accurate and resource-efficient quantum simulations.
\newpage

\section*{Declarations}

\begin{itemize}
\item Funding:
This work is based upon research supported by the
National Research Foundation of the Republic of South
Africa. 
\item Conflict of interest/Competing interests: SMP is the CEO and Co-Founder of Fraq-
tal Technologies. IJD is the COO and Co-Founder of Fraqtal Technologies.
Francesco Petruccione is the Chair of the Scientific Board and Co-Founder of
QUNOVA computing. The authors declare no other competing interests.

\item Code Availability:
The codes used to generate all the data for this research are publicly available at \url{https://github.com/fraqtal-technologies/optimising-TSPF-for-OQS-via-binary-search}.

\item Authors' contributions:
The project was conceptualized by S.M.P., I.J.D., and I.S. S.M.P. and I.J.D. developed and implemented the numerical code and prepared the initial draft of the manuscript. All Authors discussed the results and reviewed the manuscript.
\end{itemize}

\appendix

\section{Binary Search Algorithm}
\begin{algorithm}[h!]
\caption{Binary Search for Minimum $N$ that Achieves Target Precision $\epsilon$}
\label{binary_search}
\begin{algorithmic}
\State $N_{\mathrm{lower}} \gets 1$
\State $N_{\mathrm{upper}} \gets 1$

\While{$\hat{\epsilon}(t, \lambda, M, N_{\mathrm{upper}}) > \epsilon$}
    \State $N_{\mathrm{upper}} \gets 2N_{\mathrm{upper}}$
\EndWhile

\State $N_{\mathrm{mid}} \gets \big(N_{\mathrm{upper}}+N_{\mathrm{lower}}\big)//2$
\State $\hat{\epsilon} \gets \hat{\epsilon}(t, \lambda, M, N_{\mathrm{mid}})$

\While{ $N_{\mathrm{lower}} < N_{\mathrm{upper}}$ \textbf{and} $\hat{\epsilon} \ne \epsilon$}
    \If{$\hat{\epsilon} > \epsilon$}
        \State $N_{\mathrm{lower}} \gets N_{\mathrm{mid}} + 1$
        \State $N_{\mathrm{mid}} \gets \big(N_{\mathrm{upper}}+N_{\mathrm{lower}}\big)//2$
    \ElsIf {$\hat{\epsilon} < \epsilon$}
        \State $N_{\mathrm{upper}} \gets N_{\mathrm{mid}}-1$
        \State $N_{\mathrm{mid}} \gets \big(N_{\mathrm{upper}}+N_{\mathrm{lower}}\big)//2$
    \EndIf
    \State $\hat{\epsilon} \gets \hat{\epsilon}(t, \lambda, M, N_{\mathrm{mid}})$
\EndWhile
\State \Return $N_{\mathrm{mid}}$
\end{algorithmic}
\end{algorithm}

\bibliography{References}

\end{document}